\documentclass[runningheads, envsamecounter]{llncs}
\usepackage{microtype}
\usepackage{amssymb}
\usepackage{amsmath}
\usepackage{comment}
\usepackage{algorithm}
\usepackage{algorithmic}
\usepackage{paralist}

\usepackage{tikz}
\usetikzlibrary{arrows,automata,shapes,shapes.multipart}

\usepackage{xspace}

\newcommand{\defemph}[1]{\textbf{\emph{#1}}}
\newcommand{\logdcfl}{\mathsf{logDCFL}\xspace}
\newcommand{\acc}[1]{\mathsf{AC}^{#1}\xspace}
\newcommand{\ncc}[1]{\mathsf{NC}^{#1}\xspace}
\newcommand{\acd}{\mathsf{AC}^1[\mathsf{logDCFL}]\xspace}
\newcommand{\aconec}{\mathsf{AC}^1[C]\xspace}
\newcommand{\setsep}{\; | \;}
\newcommand{\fcontext}{formula context\xspace}
\newcommand{\fcontexs}{formula contexts\xspace}
\newcommand{\ap}{\textrm{AP}}
\newcommand{\ptime}{\mathsf{P}}
\newcommand{\orc}{\textrm{OR}}
\newcommand{\andc}{\textrm{AND}}
\newcommand{\notc}{\textrm{NOT}}
\newcommand{\idc}{\textrm{ID}}
\newcommand{\onec}{\textrm{ONE}}
\newcommand{\zeroc}{\textrm{ZERO}}
\newcommand{\logspace}{\textrm{L}}

\newcommand{\nn}{\mathbb{N}}
\newcommand{\rr}{\mathbb{R}}
\newcommand{\bb}{\mathbb{B}}
\newcommand{\mm}{\mathbb{M}}

\newcommand*{\ltl}{{\mathsf{LTL}}\xspace}
\newcommand*{\ltlxor}{{\mathsf{LTL+Xor}}\xspace}

\newcommand*{\mtl}{{\mathsf{MTL}}\xspace}

\newcommand*{\utl}{{\mathsf{UTL}}\xspace}

\newcommand*{\G}{\mathbin{\mathrm{G}}}

\newcommand*{\F}{\mathbin{\mathrm{F}}}
\newcommand*{\A}{\mathbin{\mathrm{A}}}
\newcommand*{\V}{\mathbin{\mathrm{V}}}
\newcommand*{\X}{\mathbin{\mathrm{X}}}
\newcommand*{\Y}{\mathbin{\mathrm{Y}}}

\newcommand*{\U}{\mathbin{\mathrm{U}}}
\newcommand*{\R}{\mathbin{\mathrm{R}}}

\newcommand*{\since}{\mathbin{\mathrm{S}}}
\newcommand*{\T}{\mathbin{\mathrm{T}}}

\newcommand{\todo}[1]{}
\newcommand{\nop}[1]{}

\begin{document}

\title{On the Complexity of Temporal-Logic\\ Path Checking\thanks{This is the full version of the paper with the same title from ICALP'14.}}
\author{Daniel Bundala \and Jo\"{e}l Ouaknine}
\institute{Department of Computer Science, University of Oxford \\ Wolfson Building, Parks Road, Oxford, OX1 3QD, UK }
\maketitle

\begin{abstract}
Given a formula in a temporal logic such as $\ltl$ or $\mtl$, a
fundamental problem is the complexity of evaluating the formula on a
given finite word. For $\ltl$, the complexity of this task was
recently shown to be in $\ncc{}$~\cite{Kuhtz:2009}. In this paper, we
present an $\ncc{}$ algorithm for $\mtl$, a quantitative (or metric)
extension of $\ltl$, and give an $\acc{1}$ algorithm for $\utl$, the
unary fragment of $\ltl$. At the time of writing, $\mtl$ is the most
expressive logic with an $\ncc{}$ path-checking algorithm, and $\utl$
is the most expressive fragment of $\ltl$ with a more efficient
path-checking algorithm than for full $\ltl$ (subject to standard
complexity-theoretic assumptions). We then establish a connection
between $\ltl$ path checking and planar circuits, which we exploit
to show that any further progress in determining the precise
complexity of $\ltl$ path checking would immediately entail more
efficient evaluation algorithms than are known for a certain class of
planar circuits.  The connection further implies that the complexity
of $\ltl$ path checking depends on the Boolean connectives
allowed: adding Boolean exclusive or yields a temporal logic with
$\ptime$-complete path-checking problem.

\end{abstract}

\section{Introduction}

One of the most fundamental problems in the fields of testing and
verification is the \emph{path-checking problem}: determine whether a
given observation\footnote{In this paper, all observations (paths,
  traces, words, etc.)\ considered are finite.} of a system satisfies
a given specification drawn from a fixed ambient logic. The complexity
of this problem plays a key role in the design and analysis of offline
monitoring and runtime verification
procedures~\cite{Finkbeiner:2004,Maler:2004}. The path-checking
problem also appears in testing~\cite{Artho:2005} and in
Monte-Carlo-based probabilistic verification~\cite{Younes:2002}.

Although the problem is simply stated, determining its precise
complexity can prove to be quite challenging. The case of $\ltl$ was
investigated more than a decade ago~\cite{Demri:2002,Markey:03}, and at the time
is was conjectured that the straightforward polynomial-time
dynamic-programming algorithm is not optimal.\footnote{The best known
  lower bound for $\ltl$ path checking is $\ncc{1}$, which crudely
  arises from the $\ncc{1}$-hardness of mere Boolean formula
  evaluation.} And indeed, using reductions to planar circuits and
tree-contraction algorithms, it was recently proved~\cite{Kuhtz:2009}
that $\ltl$ path checking allows an efficient parallel algorithm and
lies in $\ncc{}$---in fact, in $\acd$. (This seminal result was
rewarded by the ICALP~2009 best-paper award.)  More recently, this
work was extended to a very restricted metric extension of $\ltl$, in
which only temporal operators of the form $\U_{\leq b}$ are
allowed~\cite{Kuhtz:2012}.  

In this paper, we give an algorithm for full Metric Temporal Logic
($\mtl$) with the same complexity---$\acd$---
known algorithm for~$\ltl$. \nop{Our main technical contribution is a
  construction for $\U_I$ operators for an arbitrary interval
  $I$.}

We reprise the strategy, introduced in~\cite{Kuhtz:2009}, to represent
temporal operators using a special class of planar monotone circuits,
together with a generic algorithm~\cite{Chakraborty:2006} as a
subroutine to evaluate those circuits. \nop{This yields an algorithm
  in the $\acc{}$ hierarchy, precisely in $\acd$ (which is contained
  in $\acc{2}$). }Such circuits have a very special form, which led
the authors of~\cite{Kuhtz:2009} to ask whether the complexity of the
path-checking algorithm can be improved by devising specialised
circuit-evaluation algorithms. In this paper, we present evidence to
the contrary, by showing that the evaluation of circuits drawn from a
class of planar circuits studied in~\cite{Limaye:2006}
is reducible to $\ltl$ path checking; any further progress in
determining the precise complexity of the latter would therefore
immediately entail more efficient evaluation algorithms than are known
for this class of planar circuits. It is worth pointing out that
augmenting this class of planar circuits with $\notc$ gates makes the
evaluation problem $\ptime$-complete~\cite{Goldschlager:1977}. It
follows that the complexity of path checking is sensitive to
non-monotone connectives, as allowing Boolean exclusive-or in formulae
enables the evaluation of circuits from this augmented class, and is
therefore itself $\ptime$-complete.

An examination of the algorithmic constructions of~\cite{Kuhtz:2009}
shows that the most intricate parts arise in handling the Until
operator. In this paper, we show that the removal of binary operators
from the logic, yielding Unary Temporal Logic ($\utl$), leads to a
much simpler path-checking problem, enabling us to devise an $\acc{1}$
algorithm for $\utl$ path checking.

At the time of writing, our results provide (i)~the most expressive
known extension of $\ltl$ with an $\ncc{}$ path-checking algorithm
($\mtl$), (ii)~the simplest known extension of $\ltl$ with a strictly
harder path-checking problem ($\ltlxor$), and (iii)~the most
expressive known fragment of $\ltl$ with a strictly more efficient
path-checking algorithm than for full $\ltl$
($\utl$).\footnote{Subject to standard complexity-theoretic
  assumptions.}

\section{Preliminaries}
We denote Boolean true and false by $\top$ and $\bot$, respectively.
The set $\{ \bot, \top \}$ is denoted by $\bb$. A vector
$v\in\bb^n$ is \defemph{downward monotone} if $v(i+1)=\top \implies
v(i)=\top$. It is \defemph{upward monotone} if $v(i-1)=\top \implies
v(i)=\top$. A vector is \defemph{monotone} if it is upward or downward
monotone. The set of monotone vectors is denoted by $\mm$.

\textbf{Temporal Logics: } Let $\ap$ be a set of atomic propositions,
$p \in \ap$ and $I \subseteq \rr_{\geq 0}$ be an interval with
endpoints in $\nn \cup \{\infty\}$. The formulae of \defemph{Metric
  Temporal Logic} ($\mtl$) are defined recursively as
follows.
$$\varphi = p \setsep \neg p \setsep \varphi \wedge \varphi \setsep \varphi \vee \varphi \setsep \X_I \varphi \setsep \Y_I \varphi \setsep \varphi \U_I \varphi \setsep \varphi \since_I \varphi \setsep \varphi \R_I \varphi \setsep \varphi \T_I \varphi$$

All logics and results presented in this paper apply to temporal logics with past temporal operators. Note that negation is applied only to atomic
propositions. Other operators are expressible using the following
semantic equalities: $\F_I \varphi = \top \U_I \varphi, \; \G_I \varphi =
\neg \F_I \neg \varphi, \; \varphi \R_I \psi = \neg (\neg \varphi \U_I
\neg \psi)$ and $\varphi \T_I \psi = \neg (\neg \varphi \since_I \neg
\psi)$.  \defemph{Linear Temporal Logic} ($\ltl$) is the
subset of $\mtl$ in which $I$ is always $[0, \infty)$ (and is
  omitted). 
The fragment $\utl$ of $\ltl$ consists of all Boolean connectives and unary ($\X, \F, \G$) temporal operators and their past duals.

A \defemph{trace} $\pi$ over $\ap$ of length $n$ is a function
$\pi:\{1,\ldots,n\}\times\ap\to\bb$ assigning a truth value to every $p
\in \ap$ at every index. We identify $p \in \ap$
with a vector in $\bb^n$ and use $p(i) = \top$ if
$\pi(i,p)=\top$. The proposition that is true only in the interval $[i, j]$ and false otherwise is denoted by $\chi_{i,j}$, i.e., $\chi_{i,j}(k) = \top$ if $i \leq k \leq j$ and $\chi_{i,j}(k) = \bot$
otherwise. To evaluate $\mtl$ formulae on $\pi$, we further associate with $\pi$ a sequence of strictly-increasing \defemph{timestamps} $t_1 < \ldots < t_n$.

Given an $\mtl$ formula $\varphi$ and
index $1 \leq i \leq n$, the satisfaction relation $\pi, i \models
\varphi$ is defined recursively as follows.
\begin{equation*}
\begin{array}{ll}
\pi,i\models p & \text{ if } p(i) = \top \\
\pi,i\models \varphi_1 \wedge \varphi_2 & \text{ if }\pi,i\models\varphi_1 \text{ and }\pi,i\models\varphi_2 \\
\pi,i\models \varphi_1 \vee \varphi_2 & \text{ if }\pi,i\models\varphi_1 \text{ or }\pi,i\models\varphi_2 \\
\pi,i\models X_I \varphi & \text{ if } i+1<n \wedge t_{i+1}-t_i\in I \wedge \pi,i+1\models\varphi \\
\pi,i\models Y_I \varphi & \text{ if } i>1 \text{ and } t_i-t_{i-1}\in I \text{ and }\pi, i-1\models \varphi \\
\pi,i\models \varphi_1 U_I \varphi_2 &  \text{ if }\exists j \;. \; (i\leq j \leq n) \wedge
                                                  \left(\begin{array}{l}
                                                    \pi, j\models \varphi_2\\
                                                    t_j-t_i\in I \\
                                                    \forall k \;. \; i\leq k<j \implies \pi, k\models \varphi_1 \\
                                                  \end{array} \right) \\
\pi,i\models \varphi_1 S_I \varphi_2 & \text{ if }\exists j \; . \; (i\geq j\geq 1) \wedge
                                                  \left(\begin{array}{l}
                                                    \pi, j\models \varphi_2 \\
                                                    t_i-t_j\in I \\
                                                    \forall k \; . \; i\geq k>j \implies \pi, k\models \varphi_1 \\
                                                  \end{array} \right) \\
\end{array}
\end{equation*}

This paper studies the complexity of evaluating a given formula on a given~trace.


\begin{definition}
The \defemph{path-checking problem} for logic $\mathcal{L}$ is to
determine, given a trace $\pi$ and a formula $\varphi$ of
$\mathcal{L}$, whether $\pi,1\models\varphi$.
\end{definition}

Let $\varphi$ be an $\mtl$ formula. Working from the smallest subformulae and using the above definitions to tabulate the values $\pi,i\models\psi$ for every $i$ and subformula $\psi$ yields a polynomial dynamic-programming algorithm evaluating $\varphi$ on~$\pi$.

\nop{By starting from the
smallest subformulae and keeping track of values $\pi,i\models\psi$
for every subformula $\psi$ of $\varphi$ and every index, it is
possible to evaluate $\varphi$ on $\pi$ in polynomial time.
}
\begin{theorem}[\cite{Markey:03}]
The path-checking problem for $\mtl$ is in~$\ptime$.
\end{theorem}

Given a trace $\pi$ and formula $\varphi$, we represent the value of
$\varphi$ on $\pi$ as the vector $v \in \bb^n$ such that $v(i) = \top$
if and only if $\pi, i \models \varphi$. We further represent
$\ltl$ temporal operators as functions over vectors written in
infix notation. For example, $\U : \bb^n \times \bb^n \to \bb^n$ is a
function such that $(p \U q)(i) = \top$ if and only if there is $i
\leq j \leq n$ such that $q(j) = \top$ and $p(k) = \top$ for all $i
\leq k < j$.

A \defemph{\fcontext} $\varphi(X)$ is a formula
with one occurrence of a proposition replaced by a variable $X$. If $\psi(X)$ is another \fcontext then $(\varphi \circ \psi)(X)$ is the context obtained by substituting $\psi(X)$ for $X$ in $\varphi(X)$. If $q \in \ap$ is a proposition then
$\varphi(q)$ is obtained by substituting $q$ for $X$. For example,
$((p \U X ) \circ (X \since q))(r) = (p \U (X \since q))(r) = p \U (r
\since q)$. Composing \fcontexs increases the size linearly as a \fcontext contain only one occurrence of $X$.


\textbf{Circuits:}
A \defemph{Boolean circuit} $(C, \delta)$ consists of a set of \defemph{gates} $C$ and a \defemph{predecessor} function $\delta : C \to \mathcal{P}(C)$. The type of a gate is either $\orc, \andc, \notc, \idc, \linebreak \onec$ or $\zeroc$. If $c$ is of type $\tau$ and $\delta(c) = \{c_1, \ldots, c_n\}$ then we write $c = (\tau, c_1, \ldots, c_n)$\nop{ where $\tau$ is the type of $c$ and $\{c_1, \ldots, c_n\} = \delta(c)$}. If $d \in \delta(c)$ then we say $c$ \defemph{depends} on $d$ or that there is a \defemph{wire} from $d$ to $c$. The $\onec$ and $\zeroc$ gates provide constants inputs. A gate is an \defemph{input} gate if it does not have a predecessor. A gate is an \defemph{output} gate if it is not a predecessor of any other gate.\nop{The depth of a circuit is the length of a longest path in the underlying DAG.} A circuit is \defemph{monotone} if it has no $\notc$ gates. It is \defemph{planar} if the underlying DAG is planar. In this paper, all edges (wires) are straight-line segment and so a \defemph{planar embedding} is induced by a function $\gamma : C \to \rr ^ 2$ assigning a point in the plane to every gate.

A circuit is \defemph{layered} if it can be partitioned into
\defemph{layers} $C_0,\ldots, C_n$ such that each wire goes from $C_i$
to $C_{i+1}$ for some $i$. Thus, $C_0$ contains only input gates. A
layered circuit is \defemph{stratified} if all input gates appear in
$C_0$. A circuit is \defemph{upward planar} if there is a planar
embedding such that every edge monotonically increases in the upward
direction---the direction of the evaluation of $C$. A circuit is
\defemph{upward layered (stratified)} if it is both upward planar and
layered (stratified). Each layer
$C_i$ of an upward-layered circuit consists of gates $\alpha_{i,j}$
in the left-to-right ordering. \nop{Planarity means that each}Each
$\alpha_{i,j}$ depends on a contiguous block
$\alpha_{i-1,l},\ldots,\alpha_{i-1,r}$ layer below and the
wires do not cross: if $\alpha_{i,j}$ depends on
$\alpha_{i-1,q}$ and $\alpha_{i,k}$ depends on $\alpha_{i-1,r}$ then
$j \leq k \iff q \leq r$. 
Fig.~\ref{fig:transducersuphicombined} shows upward stratified
monotone circuits.

Given a circuit with one output gate, the \defemph{circuit value
  problem}, abbreviated as \defemph{CVP}, is the problem of
determining the value of the output gate.

\textbf{Complexity Classes:} The class $\logdcfl$ consists of
problems that are logspace many-one reducible to deterministic
context-free languages. Equivalently, it is the class of problems
decidable by a deterministic logspace Turing machine equipped with a stack and
terminating in polynomial time. The circuit class $\acc{i}$ for
$i\in\nn$ consists of problems decidable by polynomial-size unbounded
fan-in circuits of depth $\log^i$. All circuits in this paper are
\defemph{uniform}---can be generated by a deterministic logspace
Turing machine. Given a problem $S$ and a complexity class $C$, we write $S \in \aconec$ if there is a family of
$\acc{1}$ circuits with additional unbounded fan-in $C$-oracle gates
that decide $S$. It is known that
\begin{equation*}
\logspace \subseteq \logdcfl \subseteq \acc{1} \subseteq \acd \subseteq
\acc{2} \subseteq \cdots \subseteq \acc{i} \subseteq \acc{i+1}
\subseteq \cdots \subseteq \ptime
\end{equation*}
 and that CVP for upward-stratified
circuits is $\ptime$-complete~\cite{Goldschlager:1977}, CVP for monotone
upward-stratified circuits is in $\logdcfl$~\cite{Chakraborty:2006} and that CVP for monotone upward-layered circuits is in $\acd$~\cite{Limaye:2006}.


\textbf{Tree Contraction: }
\label{section:contraction}
Let $T=(V,E)$ be a binary tree, the tree contraction algorithm~\cite{Karp:1990} reduces $T$ to a single node using a sequence of tree contraction steps. Let $l \in T$ be a leaf, $p$ be its parent and $s$ its sibling\footnote{If $l$ does not have a sibling then we take $s$ to be a fresh node.}. A tree contraction step collapses the triple $(l,p,s)$ into a single node. Formally, a new tree $T' = (V',E')$ is obtained from $T$ as follows: $V' = V \setminus \{l, p\}$ 
\begin{equation*}
E' = \left\{
  \begin{array}{l l}
    E \setminus \{(p,l),(p,s)\}                                   & \quad \text{if $p$ is the root of $T$} \\
    (E \setminus \{(p,l),(p,s),(q,p)\})\cup\{q,s\}        & \quad \text{otherwise ($q$ is the parent of $p$)} \\
  \end{array} \right.
\end{equation*}
Note that a contraction step is local and hence multiple non-interfering contractions can be performed in parallel. A tree contraction algorithm using only~$\lceil \log{n} \rceil$ parallel steps exists~\cite{Karp:1990}. Further, this algorithm can be implemented in~$\acc{1}$. 


Let $\varphi$ be an $\ltl$ formula and $\pi$ a trace. A tree contraction algorithm evaluating $\varphi$ on $\pi$ was given in~\cite{Kuhtz:2009}. The tree $T$ used in~\cite{Kuhtz:2009} is the parse tree of $\varphi$. The leaves of $T$ correspond to the atomic propositions and the internal nodes to Boolean or temporal operators. Each contraction step $(l,p,s)$ partially evaluates the operator associated with $p$. 

For example, suppose that the formula rooted at $p$ is $\psi \U q$ where $q$ is a proposition. Even if the value of $\psi$ is unknown, we can still make some inferences. E.g., if $q(i)=\top$ then $(\psi \U q)(i)=\top$. If the last value $q(|\pi|)=\bot$ then $(\psi \U q)(|\pi|)=\bot$ and so on. The contraction step removes the nodes for $\psi$ and $\U$ and then labels the node $s$ by the partial evaluation of the function $(X \U q) \circ \psi$. It was shown in~\cite{Kuhtz:2009} how to represent, manipulate and evaluate these functions efficiently. When a subformula $\psi$ is fully collapsed into a single node then the associated function is fully evaluated and the node is labelled by the constant $(\psi(1),\ldots,\psi(|\pi|))\in\bb^{|\pi|}$. The contraction algorithm eventually reduces the tree into a single node, which is labelled by $(\varphi(1),\ldots,\varphi(|\pi|))\in\bb^{|\pi|}$.

In general, a tree-contraction algorithm can evaluate a
function $f$ on a tree; each contraction step partially
evaluating $f$ on a subtree. In this paper, the evaluation is done as
follows. Let $\mathcal{C}$ be the set of constants and $\mathcal{F}$ be a
collection, closed under composition, of admissible functions $f: \mathcal{C} \to \mathcal{C}$.

\begin{compactitem}
\item A constant $c_v\in\mathcal{C}$ is attached to every leaf $v$ of $T$. The values of $c_v$ for the initial leaves are given as a part of the input.
\item A function $f_v\in\mathcal{F}$ is attached to every node $v$ of $T$. Initially, $f_v$ is the identity function.
\item A tree contraction of $(l,p,s)$ first builds $f' \in \mathcal{F}$ (depending on $c_l$ and $p$) implementing the partial evaluation on $p$. Let $f'' = f_p \circ f'$. If $s$ is a leaf then $c_s$ is replaced by $f''(c_s)$. Otherwise, $f_s$ is replaced by $f'' \circ f_s$.
\end{compactitem}

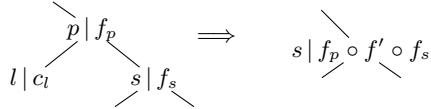
\begin{figure}
\centering
\begin{tikzpicture}[-,>=stealth',bend angle=30,auto,scale=0.55]

\tikzstyle{every state}=[rectangle, outer sep=0mm, transform shape,minimum width=10mm,minimum height=10mm]

\tikzstyle{gate}=[rectangle, outer sep=0mm, inner sep=0mm,draw=none,minimum width=5mm]
  
  \node (gl) at (0mm,0mm) [gate]{$l \,|\, c_l$};
  \node (gp) at (15mm,12mm) [gate]{$p \,|\, f_p$};
  \node (gs) at (30mm,0mm) [gate]{$s \,|\, f_s$};
  \node (gi) at (5mm,19mm) [gate]{};
  \node (gj) at (40mm,-7mm) [gate]{};
  \node (gh) at (20mm,-7mm) [gate]{};

  \node (aa) at (45mm,10mm) [gate]{$\implies$};

  \node (hs) at (80mm,7mm) [gate]{$s \,|\, f_p \circ f' \circ f_s$};
  \node (hi) at (70mm,17mm) [gate]{};
  \node (hj) at (90mm,0mm) [gate]{};
  \node (hh) at (70mm,0mm) [gate]{};
\path
  (gl) edge [] node [] {} (gp)
  (gs) edge [] node [] {} (gp)
  (gp) edge [] node [] {} (gi)
  (gs) edge [] node [] {} (gj)
  (gs) edge [] node [] {} (gh)
  
  (hs) edge [] node [] {} (hi)
  (hs) edge [] node [] {} (hj)
  (hs) edge [] node [] {} (hh)
  ;
\end{tikzpicture}
\caption{An example of a tree contraction step.}
\end{figure}

The output of the algorithm is the constant attached to the single remaining node. If each contraction step and admissible functions are in the complexity class $C$ then, by~\cite{Karp:1990}, the contraction algorithm calculating $c_{\textrm{root}}$ is in $\aconec$.


A tree contraction algorithm for $\ltl$ path checking~\cite{Kuhtz:2009} runs in~$\acd$. Constants $\mathcal{C}=\bb^n$ denote the truth values of propositions and subformulae. Functions $\mathcal{F}$ are represented by upward stratified circuits with $n$ input and $n$ output gates (\defemph{transducer circuits}), which are closed under composition~\cite{Kuhtz:2009} and their evaluation and composition is in $\logdcfl$~\cite{Barrington:1999}. For a fixed $s \in \bb^n$,~\cite{Kuhtz:2009} gives transducer circuits for $s \wedge x, s \vee x, s \U x$, and $s \R x$ as the functions of $x\in\bb^n$. In Section~\ref{mtlalgo}, we give transducer circuits for $\mtl$ temporal operators.

\section{Reduction from upward layered CVP to $\ltl$ path checking}
\label{section:reduction}

Given an upward layered monotone circuit $C$ with $n$ gates and $m$
wires we show how to build an $\ltl$ formula $\varphi$ over at
most $2n$ propositions and a trace $\pi$ of length $|\pi| \leq m$ such
that $C$ evaluates to $\top$ if and only if $\pi \models \varphi$.

Denote the layers of $C$ by $C_0, \ldots, C_k$ and the size of each $C_i$ by $n_i$. \nop{For each $0 \leq i \leq k$ and $1 \leq j \leq n_i$ }Let $\alpha_{i,j}$ be the gates in $C_i$ in the left-to-right order in the upward planar embedding of $C$. 
For each layer, we partition the trace into blocks---each of which stores the outputs of a gate in the layer. Fig.~\ref{fig:partition} shows a valid partitioning. In the figure,   gate $a$ occupies block $[1,1]$, gate $e$ occupies $[3,5]$, gate $g$ occupies $[1,7]$, etc.




\newcommand{\aij}{\alpha_{i,j}}

In general, a valid partitioning consists of a trace $\pi$ and intervals $v(i,j)$ associated with each gate $\aij$ such that $v(i,j)$ overlaps precisely with the blocks of the gates the gate $\aij$ depends on. Formally,


\begin{compactitem}
\item intervals $v(i,1), v(i, 2), \ldots, v(i,n_i)$ are disjoint and partition $[1,|\pi|]$ for every~$i$,
\item if $\alpha_{i+1,j}$ depends on $\alpha_{i,p},\alpha_{i,p+1},\ldots,\alpha_{i,q}$ then $v(i+1,j)\subseteq \cup_{r=p,\ldots,q} v(i,r)$ and $v(i+1,j)$ overlaps with each $v(i,r)$ for $p\leq r\leq q$,
\end{compactitem}

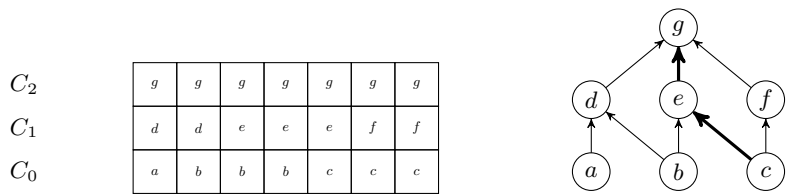
\begin{figure}
\centering
\begin{tikzpicture}[->,>=stealth',bend angle=30,auto,scale=0.58]

\tikzstyle{every state}=[rectangle, outer sep=0mm, transform shape,minimum width=10mm,minimum height=10mm]

  \node (u) at (-30mm,0mm) [draw=none]{$C_0$};
  \node (u0) at (0mm,0mm) [state]{$a$};
  \node (u1) at (10mm,0mm) [state]{$b$};
  \node (u2) at (20mm,0mm) [state]{$b$};
  \node (u2) at (30mm,0mm) [state]{$b$};
  \node (u2) at (40mm,0mm) [state]{$c$};
  \node (u2) at (50mm,0mm) [state]{$c$};
  \node (u3) at (60mm,0mm) [state]{$c$};

  \node (v) at (-30mm,10mm) [draw=none]{$C_1$};
  \node (v0) at (0mm,10mm) [state]{$d$};
  \node (v1) at (10mm,10mm) [state]{$d$};
  \node (v2) at (20mm,10mm) [state]{$e$};
  \node (v2) at (30mm,10mm) [state]{$e$};
  \node (v2) at (40mm,10mm) [state]{$e$};
  \node (v2) at (50mm,10mm) [state]{$f$};
  \node (v3) at (60mm,10mm) [state]{$f$};

  \node (w) at (-30mm,20mm) [draw=none]{$C_2$};
  \node (w0) at (0mm,20mm) [state]{$g$};
  \node (w1) at (10mm,20mm) [state]{$g$};
  \node (w2) at (20mm,20mm) [state]{$g$};
  \node (w2) at (30mm,20mm) [state]{$g$};
  \node (w2) at (40mm,20mm) [state]{$g$};
  \node (w2) at (50mm,20mm) [state]{$g$};
  \node (w3) at (60mm,20mm) [state]{$g$};

\tikzstyle{gate}=[circle, outer sep=0mm, inner sep=0mm,draw,minimum width=5mm]
  
  \node (ga) at (100mm,0mm) [gate]{$a$};
  \node (gb) at (120mm,0mm) [gate]{$b$};
  \node (gc) at (140mm,0mm) [gate]{$c$};
  \node (gd) at (100mm,16.5mm) [gate]{$d$};
  \node (ge) at (120mm,16.5mm) [gate]{$e$};
  \node (gf) at (140mm,16.5mm) [gate]{$f$};
  \node (gg) at (120mm,33mm) [gate]{$g$};
\path
  (ga) edge [] node [] {} (gd)
  (gb) edge [] node [] {} (gd)
  (gb) edge [] node [] {} (ge)
  (gc) edge [very thick] node [] {} (ge)
  (gc) edge [] node [] {} (gf)
  (gd) edge [] node [] {} (gg)
  (ge) edge [very thick] node [] {} (gg)
  (gf) edge [] node [] {} (gg)
  ;
\end{tikzpicture}
\caption{An upward layered circuit (on the right) with its partition (on the left). The path $\pi$ for the gate labelled $e$ is highlighted.}
\label{fig:partition}
\end{figure}

Suppose we are given a valid partitioning. Then for $i > 0$ and every $1 \leq j \leq n_i$ we build a \fcontext $\varphi_{i,j}$ mimicking the evaluation of the gate $\alpha_{i,j}$.\nop{\footnote{E.g., assuming gate $e$ in Fig.~\ref{fig:partition}} is an $\orc$ gate, we have $\varphi_{1,2} = \chi_{4,5} \since (\chi_{3,4} \U X)$.} 

For example, suppose that the gate $e$ in Fig.~\ref{fig:partition} is an $\orc$ gate and the values of the block in the first layer is $r = (a,b,b,b,c,c,c)\in\bb^7$ for some $a,b,c \in \bb$. Recall that $(\varphi \U \psi)(i) = \psi(i) \vee (\varphi(i) \wedge (\varphi \U \psi)(i+1))$. Hence, if $\varphi(i) = \bot$ then $(\varphi \U \psi)(i) = \psi(i)$ and if $\varphi(i) = \top$ then $(\varphi \U \psi)(i) = \psi(i) \vee (\varphi \U \psi)(i+1)$. Further recall that $\chi_{i,j}$ is a proposition that is true on $[i,j]$ and false otherwise. Hence, $(\chi_{3,4} \U r)(1) = a, (\chi_{3,4} \U r)(2) = b$ and $(\chi_{3,4} \U r)(5,6,7) = c$. Also, $(\chi_{3,4} \U r)(4) = r(4) \vee (\chi_{3,4} \U r)(5) = b \vee c$. Finally, $(\chi_{3,4} \U r)(3) = r(3) \vee (\chi_{3,4} \U r)(4) = b \vee (b \vee c) = b \vee c$. So $\chi_{3,4} \U r = (a, b, b \vee c, b \vee c, c, c, c)$. Performing a similar calculation backwards, we get $\chi_{4,5} \since (\chi_{3,4} \U r) = (a, b, b \vee c, b \vee c, b \vee c, c, c)$ which gives the value of block $e$ in Fig.~\ref{fig:partition} and leaves other blocks unchanged.

Denote the type of $\alpha_{i,j}$ by $\tau$ and the left and the right endpoint of $v(i,j)$ by $l$ and $r$, respectively. Then $\varphi_{i,j}$ is constructed as follows:

\begin{compactitem}
\item If $\tau = \onec$ then $\varphi_{i,j}(X) = \chi_{l,r} \vee X$.
\item If $\tau = \zeroc$ then $\varphi_{i,j}(X) = (\neg \chi_{l,r}) \wedge X$.
\item If $\tau = \idc$ then $\varphi_{i,j}(X) = X$.
\item If $\tau = \orc$ then $\varphi_{i,j}(X) = \chi_{l+1, r} \since (\chi_{l, r-1} \U X)$.
\item If $\tau = \andc$ then $\varphi_{i,j}(X) = \chi_{l+1, r} \T (\chi_{l, r-1} \R X)$.
\end{compactitem}

It can be shown that the \fcontext $\varphi_{i,j}$ updates the block $v(i,j)$ and leaves the other blocks unchanged. Hence, the \fcontext $\psi_i(X) = \varphi_{i,1} \circ \varphi_{i,2} \circ \cdots \circ \varphi_{i,n_i}$ evaluates the i-th layer $C_i$ of $C$. 

Formally, for each layer $C_i$ let $r_i\in\bb^n$ be a proposition such that $r_i(k) = \top$ if $k \in v(i,j)$ for some $j$ and $\alpha_{i,j}$ evaluates to $\top$ and $r_i(k)=\bot$, otherwise. Then, the formula $\varphi = (\psi_k \circ \psi_{k-1} \circ \cdots \circ \psi_1)(r_0)$ computes the output of the circuit.


\begin{lemma}
\label{lemma:phi}
Let $\psi_i, \varphi$ be as above. Then $\psi_i(r_{i-1}) = r_i$ and $\varphi(r_0)(1) = \top$ if and only if $C$ evaluates to $\top$. Moreover, $\varphi$ can be built in $\logspace$.
\end{lemma}

Finally, we show how to devise $v(i,j)$'s -- the partitioning of the trace. Without loss of generality, connecting to a gate in the previous layer if necessary, we assume that all $\onec$ and $\zeroc$ gates not in $C_0$ have at least one predecessor.

Given a gate $\aij$ there is unique rightmost gate in the layer $C_{i+1}$ that $\aij$ is connected to by a wire. Now, start at $\aij$ and take the rightmost wires until the sink is reached. Denote the traversed path by $\pi_u$.
Similarly, there is unique rightmost gate in the layer $C_{i-1}$ that $\aij$ is connected to by a wire. Start at $\aij$ and take the rightmost wires going down until a gate in $C_0$ is reached. Denote the traversed path by $\pi_d$. Let $\pi$ be the concatenation of $\pi_d$ and $\pi_u$. (See Fig.~\ref{fig:partition})

Let $k_{i,j}$ be the number of wires to the left of $\pi$. A wire from $\alpha_{i,j}$ to $\alpha_{i+1,k}$ is to the left of the wire from $\alpha_{i,a}$ to $\alpha_{i+1,b}$ if $j < a$ or $k < b$. We store the output of gate $\aij$ in the block $v(i,j):=[k_{i,j-1} + 1,k_{i,j} + 1]$. We use $k_{i,0}=0$. 

Fig.~\ref{fig:partition} shows a circuit and the partitioning obtained by the above procedure. The rightmost wire going up and down from $e$ are $e \to g$ and $c \to e$, respectively. Thus, $\pi_u = e \to g$ and $\pi_d= c \to e$. The path $\pi = c \to e \to g$ is highlighted in the figure. Four wires $a \to d, b \to d, b \to e, d \to g$ are to the left of $\pi$. We associate the block $[3,5]$ with gate $e$. All blocks, grouped by layers, are shown in Fig.~\ref{fig:partition}.

The following lemma summarises the important properties of $k_{i,j}$'s.
\begin{lemma}
\label{lemma:kijs}
Let $k_{i,j}$'s and $v(i,j)$'s be as above. Then the following hold:
\begin{compactitem}
\item $k_{i,j-1}<k_{i,j}$ for every $i$ and $j$,
\item $k_{i,n_i}=k_{j,n_j}$ for every $i$ and $j$,
\item $k_{i,n_i}\leq m$ for every $i$,
\item for every $i$ and $j=1,\ldots,n_i$ the intervals $v(i,j)$'s partition $[1,k_{i,n_i}]$,
\item if $\alpha_{i+1,j}$ depends on $\alpha_{i,p},\alpha_{i,p+1},\ldots,\alpha_{i,q}$ then $v(i+1,j)\subseteq \cup_{r=p,\ldots,q} v(i,r)$ and $v(i+1,j)$ overlaps with each $v(i,r)$ for $p\leq r\leq q$,
\item each $k_{i,j}$ can be computed in $\logspace$.
\end{compactitem}
\end{lemma}
This finishes the reduction from upward-layered CVP to $\ltl$ path checking.  It was shown in~\cite{Kuhtz:2009} that the latter is in $\acd$. Therefore:

\begin{theorem}
\label{thm:upwardlayered}
The CVP for upward-layered monotone circuits is in $\acd$.
\end{theorem}

An alternative proof of Theorem~\ref{thm:upwardlayered} already appeared in~\cite{Limaye:2006}. Moreover, the relationship shows that any improvement in $\ltl$ path checking would entail an improvement in the evaluation of upward-layered monotone circuits. 

The above reduction assumes the monotonicity of the input circuit. However, if the target logic $\ltl$ is extended to include binary exclusive or (xor) as a connective, then evaluating $\notc$ gates becomes possible using $\varphi_{i,j}(X) = \chi_{l,r} \oplus X$ as a \fcontext for $\notc$ gate $\alpha_{i,j}$. Noting that CVP is P-complete for general (non-monotone) upward stratified circuits~\cite{Goldschlager:1977}, we have the following:

\begin{theorem}
$\ltlxor$ path checking is $\ptime$-complete.
\end{theorem}

Thus, the complexity of $\ltl$ path checking depends on the monotonicity of the Boolean connectives present in the formula.

\section{$\mtl$ path checking is efficiently parallelisable}
\label{mtlalgo}

We now show how the tree-contraction method of~\cite{Kuhtz:2009} extends to full $\mtl$;~giving an $\acd$ path-checking algorithm for $\mtl$.\nop{An $\acd$ path-checking algorithm was obtained previously~\cite{Kuhtz:2012} for a very restrictive fragment of $\mtl$.} By~\cite{Kuhtz:2009}, summarised in~Section~\ref{section:contraction}, it suffices to give upward stratified transducer circuits for $\U_I$ and its~duals.

Let $\pi$ be the input trace with (floating-point) timestamps $t_1,\ldots,t_n$. Fix an interval $I$ and consider the $\U_I$ operator. We now describe a dynamic-programming approach that yields planar circuits calculating $(\psi_1 \U_I \psi_2)(i)$. For $i \neq j$ the values $(\psi_1 \U_I \psi_2)(i)$ and $(\psi_1 \U_I \psi_2)(j)$ depend on the values of subformulae in some future intervals. In general, these intervals overlap and so naive constructions of transducer circuits are not planar.  See Fig.~\ref{fig:transducersuphicombined} for the kind of circuits we build.

Recall, that the tree contraction is applied only to a leaf, its parent and its sibling. Let $s\in\bb^n$ be a vector. We need to construct only circuits for $s\U_I \varphi$ and $\varphi\U_I s$ for known $s$. First consider the case $s\U_I\varphi$. (see left part of Fig.~\ref{fig:transducersuphicombined})

For index $1\leq i\leq n$ the formula $(s \U_I \varphi)(i)$ is true if there is $j \geq i$ such that $t_j \in t_i + I$ and $\varphi(j) = \top$ and $s(k) = \top$ for all $i \leq k < j$. So let $T_i = \{j \setsep t_j \in t_i + I\}$ be the set of  indices in $t_i + I$. If $T_i=\emptyset$ then $(s \U_I \varphi)(i) = \bot$. 

Otherwise, let $\textrm{first}(i) = \min T_i$ and $\textrm{last}(i) = \max T_i$ be the first and the last index in the interval $t_i+I$, respectively. So $(s \U_I \varphi)(i)$ is true if there exists $\textrm{first}(i) \leq j \leq \textrm{last}(i)$ such that $\varphi(j) = \top$ and $s(k) = \top$ for all $i \leq k < j$.

Now, the value of  $s$ is known. So let $\textrm{seg}(i)=\min\{j \setsep j\geq i \wedge s(j)=\bot \}$ be the first index no smaller than $i$ such that $s(j)$ evaluates to false, i.e., $s(j)$ is true from $i$ to $\textrm{seg}(i) - 1$. Thus, $(s \U_I \varphi)(i)$ is true if there exists $\textrm{first}(i) \leq j \leq \textrm{last}(i)$ such that $\varphi(j) = \top$ and $ j \leq \textrm{seg}(i)$.  So take $L_i = \textrm{first}(i)$ and $R_i = \min(\textrm{last}(i), \textrm{seg}(i))$. Then $(s \U_I \varphi)(i)$ is true if $\bigvee_{L_i\leq j\leq R_i} \varphi(j)$ is true.
\begin{figure}
\centering
\begin{tikzpicture}[->,>=stealth',bend angle=30,auto,scale=0.44]
\tikzstyle{every state}=[rectangle, outer sep=0mm, transform shape,minimum width=10mm,minimum height=10mm]

\tikzstyle{gate}=[rectangle, outer sep=1mm, inner sep=0mm,draw=none,minimum width=5mm]
\tikzstyle{lab}=[rectangle, outer sep=1mm, inner sep=0mm,draw=none,minimum width=5mm,font=\scriptsize]
  
  \node (p1) at (0,0) [gate]{$\varphi_1$};
  \node (p2) at (2,0) [gate]{$\varphi_2$};
  \node (p3) at (4,0) [gate]{$\varphi_3$};
  \node (p4) at (6,0) [gate]{$\varphi_4$};
  \node (p5) at (8,0) [gate]{$\varphi_5$};
  \node (p6) at (10,0) [gate]{$\varphi_6$};
  \node (p7) at (12,0) [gate]{$\varphi_7$};
  
  \node (u1) at (0,-1) [lab]{$t: 1$};
  \node (v1) at (0,-2) [lab]{$s:\bot$};
  \node (u2) at (2,-1) [lab]{$2$};
  \node (v2) at (2,-2) [lab]{$\top$};
  \node (u3) at (4,-1) [lab]{$3$};
  \node (v3) at (4,-2) [lab]{$\top$};
  \node (u4) at (6,-1) [lab]{$4$};
  \node (v4) at (6,-2) [lab]{$\top$};
  \node (u5) at (8,-1) [lab]{$5$};
  \node (v5) at (8,-2) [lab]{$\top$};
  \node (u6) at (10,-1) [lab]{$6$};
  \node (v6) at (10,-2) [lab]{$\top$};
  \node (u7) at (12,-1) [lab]{$8.5$};
  \node (v7) at (12,-2) [lab]{$\bot$};

  \node (s34) at (4,2) [gate]{$\vee$};
  \node (s35) at (4,4) [gate]{$\vee$};
  \node (s36) at (4,6) [gate]{$\vee$};
  \node (s45) at (6,2) [gate]{$\vee$};
  \node (s46) at (6,4) [gate]{$\vee$};
  \node (s56) at (8,2) [gate]{$\vee$};
  \node (s57) at (8,4) [gate]{$\vee$};
  \node (s67) at (10,2) [gate]{$\vee$};
  \node (o1) at (0,8) [gate]{$o_1$};
  \node (o2) at (2,8) [gate]{$o_2$};
  \node (o3) at (4,8) [gate]{$o_3$};
  \node (o4) at (6,8) [gate]{$o_4$};
  \node (o5) at (8,8) [gate]{$o_5$};
  \node (o6) at (10,8) [gate]{$o_6$};
  \node (o7) at (12,8) [gate]{$o_7$};
  
  \node (b1) at (0,7) [gate]{$\bot$};
  \node (b7) at (12,7) [gate]{$\bot$};
  \node (xx) at (13,-3) [gate]{};
  \node (yy) at (13,9) [gate]{};
\path[-]
  (xx) edge [dashed] node [] {} (yy)
  ;
\path
  (p3) edge [] node [] {} (s34)
  (p4) edge [] node [] {} (s34)
  (p4) edge [] node [] {} (s45)
  (p5) edge [] node [] {} (s45)
  (p5) edge [] node [] {} (s56)
  (p6) edge [] node [] {} (s56)
  (p6) edge [] node [] {} (s67)
  (p7) edge [] node [] {} (s67)
  
  (s34) edge [] node [] {} (s35)
  (s45) edge [] node [] {} (s35)
  (s45) edge [] node [] {} (s46)
  (s56) edge [] node [] {} (s46)
  (s56) edge [] node [] {} (s57)
  (s67) edge [] node [] {} (s57)

  (s35) edge [] node [] {} (s36)
  (s46) edge [] node [] {} (s36)
  
  (s36) edge [] node [] {} (o2)
  (s46) edge [] node [] {} (o3)
  (s57) edge [] node [] {} (o4)
  (s67) edge [] node [] {} (o5)
  (p7) edge [] node [] {} (o6)
  
  (b1) edge [] node [] {} (o1)
  (b7) edge [] node [] {} (o7)
  
  ;

  \node (ap1) at (14,0) [gate]{$\varphi_1$};
  \node (ap2) at (16,0) [gate]{$\varphi_2$};
  \node (ap3) at (18,0) [gate]{$\varphi_3$};
  \node (ap4) at (20,0) [gate]{$\varphi_4$};
  \node (ap5) at (22,0) [gate]{$\varphi_5$};
  \node (ap6) at (24,0) [gate]{$\varphi_6$};
  \node (ap7) at (26,0) [gate]{$\varphi_7$};
  
  \node (au1) at (14,-1) [lab]{$t:1$};
  \node (av1) at (14,-2) [lab]{$s:\bot$};
  \node (au2) at (16,-1) [lab]{$2$};
  \node (av2) at (16,-2) [lab]{$\bot$};
  \node (au3) at (18,-1) [lab]{$3$};
  \node (av3) at (18,-2) [lab]{$\bot$};
  \node (au4) at (20,-1) [lab]{$3.5$};
  \node (av4) at (20,-2) [lab]{$\bot$};
  \node (au5) at (22,-1) [lab]{$3.8$};
  \node (av5) at (22,-2) [lab]{$\top$};
  \node (au6) at (24,-1) [lab]{$4$};
  \node (av6) at (24,-2) [lab]{$\bot$};
  \node (au7) at (26,-1) [lab]{$4.5$};
  \node (av7) at (26,-2) [lab]{$\top$};

  \node (as12) at (14,2) [gate]{$\wedge$};
  \node (as13) at (14,4) [gate]{$\wedge$};
  \node (as14) at (14,6) [gate]{$\wedge$};
  \node (as23) at (16,2) [gate]{$\wedge$};
  \node (as24) at (16,4) [gate]{$\wedge$};
  \node (as34) at (18,2) [gate]{$\wedge$};
  \node (as35) at (18,4) [gate]{$\wedge$};
  \node (as36) at (18,6) [gate]{$\wedge$};
  \node (as45) at (20,2) [gate]{$\wedge$};
  \node (as46) at (20,4) [gate]{$\wedge$};
  \node (as56) at (22,2) [gate]{$\wedge$};
  \node (as57) at (22,4) [gate]{$\wedge$};
  \node (as67) at (24,2) [gate]{$\wedge$};
  \node (ao1) at (14,8) [gate]{$o_1$};
  \node (ao2) at (16,8) [gate]{$o_2$};
  \node (ao3) at (18,8) [gate]{$o_3$};
  \node (ao4) at (20,8) [gate]{$o_4$};
  \node (ao5) at (22,8) [gate]{$o_5$};
  \node (ao6) at (24,8) [gate]{$o_6$};
  \node (ao7) at (26,8) [gate]{$o_7$};
  
  \node (at7) at (26,7) [gate]{$\bot$};
\path
  (ap1) edge [] node [] {} (as12)
  (ap2) edge [] node [] {} (as23)
  (ap3) edge [] node [] {} (as34)
  (ap4) edge [] node [] {} (as45)
  (ap5) edge [] node [] {} (as56)
  (ap6) edge [] node [] {} (as67)
  (ap2) edge [] node [] {} (as12)
  (ap3) edge [] node [] {} (as23)
  (ap4) edge [] node [] {} (as34)
  (ap5) edge [] node [] {} (as45)
  (ap6) edge [] node [] {} (as56)
  (ap7) edge [] node [] {} (as67)
  
  (as12) edge [] node [] {} (as13)
  (as23) edge [] node [] {} (as13)
  (as23) edge [] node [] {} (as24)
  (as34) edge [] node [] {} (as24)
  (as34) edge [] node [] {} (as35)
  (as45) edge [] node [] {} (as35)
  (as45) edge [] node [] {} (as46)
  (as56) edge [] node [] {} (as46)
  (as56) edge [] node [] {} (as57)
  (as67) edge [] node [] {} (as57)

  (as13) edge [] node [] {} (as14)
  (as24) edge [] node [] {} (as14)
  
  (as35) edge [] node [] {} (as36)
  (as46) edge [] node [] {} (as36)
  
  (as14) edge [] node [] {} (ao1)
  (as24) edge [] node [] {} (ao2)
  (as36) edge [] node [] {} (ao3)
  (as46) edge [] node [] {} (ao4)
  (as57) edge [] node [] {} (ao5)
  (as67) edge [] node [] {} (ao6)

  (at7) edge [] node [] {} (ao7)
  
  ;
\end{tikzpicture}
\caption{Transducer circuits for $s \U_{[1,5]}\varphi$ and $\varphi \U_{[1,5]} s$. The first line below the circuits are timestamps, the second row are values of $s$. Note that different timestamps and $s$ are used in the two examples. The inputs and the outputs of the circuits are denoted $\varphi_i$ and $o_i$ respectively.}
\label{fig:transducersuphicombined}
\end{figure}
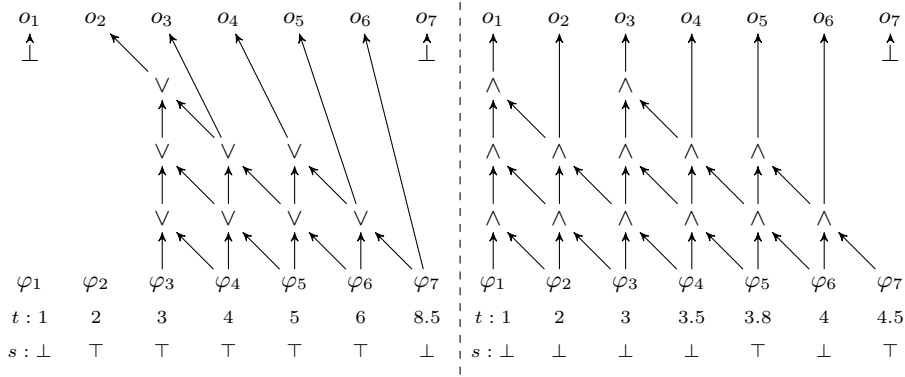

To build the circuits, we formalise the intuition from the left half of Fig.~\ref{fig:transducersuphicombined}. The circuit $C$ consists of internal gates $d_{p,q}$ and output gates $o_i$ for each $1 \leq i \leq n$. Each internal gate $d_{p,q}$ calculates $\varphi_p \vee \cdots \vee \varphi_q$. Precisely, $d_{p,q}$ is present in the circuit if there is an $i$ such that $L_i \leq p \leq q \leq R_i$. If $p=q$ then $l(d_{p,q}) = (\idc, \varphi_p)$. Otherwise,  $l(d_{p,q}) = (\orc, d_{p,q-1}, d_{p+1,q})$.


For the output gates, we define $o_i$ so that $o_i = \bigvee_{L_i\leq j\leq R_i} \varphi(j) = (s \U_I \varphi)(i)$. Specifically, if $T_i = \emptyset$  then we set $l(o_i)=\bot$, otherwise, $l(o_i) = (\idc, d_{L_i, R_i})$.

An embedding $\gamma : C \to \rr^2$ for the circuit $C$ is $\gamma(o_i) = (i, n)$, $\gamma(\varphi_i) = (i, 0)$ and $\gamma(d_{p,q})=(p, q-p+1)$. Observe that $L_i\leq L_{i+1}$ and $R_i\leq R_{i+1}$. Hence, it cannot happen that $L_i<L_j\leq R_j<R_i$ for some $i$ and $j$. So the intervals may overlap but never is one properly contained in another. This ensures that the embedding is planar.

Finally, note that it is possible to compute $L_i$ and $R_i$ for every $i$ in logarithmic space\nop{\footnote{$t_j - t_i \geq a$ can be checked by first comparing the largest bits. If they are equal then by comparing the second largest bits and so on.}}. Hence, the circuit construction can be carried out in logarithmic space. 


\begin{lemma}
\label{lemma:circuit}
Let $p$ be any proposition. For each $i$, set the input $\varphi_i$ of the circuit to $p(i)$. Then for each $j$, the value of $o_j$ is true if and only if $(s \U_I p)(j)$ is true. 
\end{lemma}

We now give an analogous derivation and circuit construction for $\varphi \U_I s$. See the right side of Fig.~\ref{fig:transducersuphicombined} for an example of a resulting circuit. 

For index $1\leq i\leq n$ the formula $(\varphi \U_I s)(i)$ is true if there exists $j \geq i$ such that $t_j \in t_i + I$ and $s(j)=\top$ and $\varphi(k)=\top$ for all $i \leq k < j$. Since $s$ is known, we choose the first possible $j$. So let $\textrm{limit}(i) = \min \{ j \setsep \textrm{first}(i) \leq j \leq \textrm{last}(i) \wedge s(j)=\top \}$ be the first $j$ in the interval $t_i + I$ such that $s(j)$ is true.

If there is no such index then $(\varphi \U_I s)(i) = \bot$. Otherwise, $(\varphi \U_I s)(i)$ is true if $\varphi(k)=\top$ for all $i \leq k < \textrm{limit}(i)$. That is, $(\varphi \U_I s)(i) = \bigwedge_{i\leq j<\textrm{limit}(i)}\varphi_j$.

Now, the circuit $C$ (see right half of Fig.~\ref{fig:transducersuphicombined}) consists of gates $c_{p,q}$ calculating $\varphi_p \wedge \cdots \wedge \varphi_q$ and output gates $o_i$ for $i=1\ldots n$. The gate $c_{p,q}$ is present in $C$ if there is $i$ such that $i \leq p \leq q < \textrm{limit}(i)$. If $p=q$ then $l(c_{p,q}) = (\idc, \varphi_p)$. Otherwise, $l(c_{p,q}) = (\andc, c_{p,q-1}, c_{p+1,q})$.

For output, we set $o_i$ so that  $o_i= \bigwedge_{i\leq j<\textrm{limit}(i)}\varphi_j = (\varphi \U_I s)(i)$. If $\textrm{limit}(i)=\infty$ then $l(o_i)=\bot$, if $\textrm{limit}(i)=i$ then $l(o_i)=\top$ and else $l(o_i)=(\idc, c_{i,\textrm{limit}(i)-1})$.


The embedding $\gamma:C\to\rr^2$ of the circuit $C$ is the same as above, $\gamma(o_i) = (i, n)$, $\gamma(\varphi_i) = (i, 0)$ and $\gamma(c_{p,q})=(p, q-p+1)$. Since, $i < j$ implies $\textrm{limit}(i) \leq \textrm{limit}(j)$, the embedding is planar.

This finishes the construction of circuits for $\U_I$. Circuits for the dual operators of $\U_I$ are obtained either by dualising $\orc$ and $\andc$ gates (Release operator), by performing the construction backwards in time (Since operator) or both (Trigger operator). Therefore,

\begin{theorem}
$\mtl$ path checking is in $\acd$.
\end{theorem}

A considerably weaker result appeared in \cite{Kuhtz:2012}, where
the authors gave circuits and an $\acd$ algorithm only for a fragment
of $\mtl$ interpreted over traces with integral timestamps $t_i =
i$ and intervals of the form $[0, a]$ for $a \in \nn$. 

\section{$\utl$}

The most complicated circuits in the $\ltl$ path-checking algorithm~\cite{Kuhtz:2009} correspond to $s \U \psi$ and $\psi \U s$ formulae. As in the case of $\mtl$, the circuits are also not uniform but depend on $s$.
In this section, we devise an $\acc{1}$ tree-contraction algorithm for $\utl$---the fragment of $\ltl$ obtained by omitting binary temporal operators. The algorithm works even if XOR is allowed and is based on the analysis of functions arising in the tree contraction algorithm applied to $\utl$ formulae. First consider the future-only fragment of $\utl$.

Let $p\in\bb^n$ be any proposition. If $p(i) = \bot$ for every $i$ then $(\F p)(i)=\bot$ for every $i$. Otherwise, let $i$ be the largest index such that $p(i)=\top$. Then, $(\F p)(j) = \top$ for all $j\leq i$. By construction, $p(k) = \bot$ for all $k > i$. Hence, $(\F p)(k) = \bot$ for all $k>i$. Thus, $\F p$ is downward monotone and depends only on the largest $i$ with $p(i)=\top$. In particular, only $n+1$ possible values exist for $\F p$.

Similarly, let $t$ be the largest index such that $p(t) =\bot$. Then $p(j) = \top$ for all $j > t$. Hence $(\G p)(j) = \top$ for all $j > t$. Since $p(t) = \bot$ we have $(\G p)(k) = \bot$ for all $k \leq t$. Thus, $\G p$ is upward monotone and depends only on the largest $t$ with $p(t)=\bot$. In particular, only $n+1$ possible values exist for $\G p$.

So for any formula $\psi$ the value of $\F \circ \psi$ or $\G \circ \psi$ is a monotone vector---of which there are only $2n$ many. Hence for any \fcontext $\varphi(X)$, the \fcontexs $\varphi \circ (\F X)$ and $\varphi \circ (\G X)$ can be represented as $g \circ F$ or $g \circ G$ where $g : \mm \to \bb^n$ is a \defemph{function with monotone domain}. Since $|\mm|=O(n)$, enumerating all outputs of $g$ explicitly requires only $|g| = O(n^2)$ space. Similar results  hold for the past equivalents of $\G$ and $\F$.


Now, Boolean operators are applied componentwise and obey the usual identities: $\bot \wedge p = \bot, \top \wedge p = p, \bot \vee p = p, \top \vee p = \top, \bot \oplus p = p$ and $\top \oplus p = \neg p$. Therefore, to represent partial evaluation of conjunction ($p \wedge X, x \wedge X$), disjunction ($p \vee X, X \vee p$) and xor ($p \oplus X, X \oplus p$) it suffices to keep track whether each component is $\bot, \top$ or equal to the original or the negation of the value in~$X$.


Furthermore, Next ($\X p$) and Yesterday ($\Y p$) temporal operators shift $p$ by $1$ and $-1$, respectively. Let $m$ be the size of the input formula. The last two paragraphs motivate the definition of filters: let $v \in \{\bot, \top, \idc, \notc\} ^ n$ and $k \in [-m, m]$ satisfy $v(i) \in \bb$ if $i + k \not\in \{1, \ldots, n\}$. Then a \defemph{filter with offset $k$ and pattern $v$} is the function $f_{v, k} : \bb^n\to\bb^n$ such that

$$
f_{v,k}(p)(i) = \left\{
  \begin{array}{l l}
    \bot        & \quad \text{if $v(i) = \bot$} \\
    \top        & \quad \text{if $v(i) = \top$} \\
    p(i+k)      & \quad \text{if $v(i) = \idc$} \\
    \neg p(i+k) & \quad \text{if $v(i) = \notc$} \\
  \end{array} \right.
$$  

\nop{For each component an offset keeps track of whether it is $\bot, \top$, negated or the original value. }The identity function as well as the partial evaluation of conjunction, disjunction, and xor are expressible as filters with offset $0$. 
Temporal operators Next and Yesterday are identity filters with offsets $1$ and $-1$, respectively. Note that filters are closed under composition.

Storing $v$ explicitly and $k$ in unary requires $O(n + |\varphi|)$ bits per filter.\nop{ If the offset is stored in unary and $2$ bits are used for each entry of $v$ then $2n + 2|\varphi|+1$ bits are needed to represent the filter $f_{v,k}$. } By fully expanding the definition, we can evaluate and compose two filters in $\acc{0}$. Moreover, if $g:\mm\to\bb^n$ is a function with monotone domain then $(f_{v,k} \circ g): \mm\to\bb^n$ is also a function with monotone domain and the composition in~$\acc{0}$.

\begin{lemma}
\label{lemma:utlcomposition}
There are uniform $\acc{0}$ circuits calculating $f_{v, k} \circ f_{v',k'}$ and $f_{v,k}(p)(i)$ and $f_{v,k} \circ g$ and $\F \circ g$ and $\G \circ g$, where $f$'s are filters and $g$ is a function with monotone domain.
\end{lemma}


We represent the functions arising in the tree-contraction algorithm as follows. If the contracted subtree $S$ does not contain $\F$ or $\G$ operators then it is representable by a filter. If it contains $\F$ or $\G$ then let $T$ be the first such occurrence. Then the segment from the leaves to $T$ is representable by a filter and the segment above $T$ is representable by a function with monotone domain. Thus, the function $h$ associated with $S$ can be represented as:
$$
h = \left\{
  \begin{array}{l l}
    \textrm{filter}        & \quad \text{no temporal operator} \\
    f \circ T \circ \textrm{filter}        & \quad \text{$T$ is the first temporal operator; $f: \mm \to \bb^n$} \\
  \end{array} \right.
$$  
\nop{That is, the function is either a filter or a filter followed by a unary temporal operator and a function with monotone domain. The choice corresponds to the presence of a temporal operator in the contracted subtree. Let $h$ be a function of this form. }
Now, if the contracted node is a Boolean connective, $\X$ or $\Y$ then we calculate $f_{v,k} \circ h$ for an appropriate filter. If the contracted node is $\F$ or $\G$ then we calculate $\F \circ h$ or $\G \circ h$. In either case, the resulting function is representable using the above format\nop{ of Fig.~\ref{fig:utlfunction}}. Moreover, by Lemma~\ref{lemma:utlcomposition}, the composition is in $\acc{0}$. Hence, the complexity of the tree contraction algorithm is $\acc{1}[\acc{0}]=\acc{1}$.
\nop{
\begin{figure}
\centering
\begin{tikzpicture}[->,>=stealth',bend angle=30,auto,scale=0.6]

\tikzstyle{every state}=[inner ysep=-14mm, rectangle, outer sep=0mm, transform shape, minimum width=30mm]

  \node (s0) at (0,0) [state, initial below, initial text=]{filter};
  \node (s1) at (2,1.5) [state]{id};
  \node (s2) at (-2,1.5) [state]{$\F / \G / \A / \V$};
  \node (s3) at (-2,3) [state]{$f:\mm\to\bb^n$};
  \node (s4) at (-2,4.5) [state, draw=none]{};

\path
  (s0) edge [] node [above] {} (s1)
  (s0) edge [] node [above] {} (s2)
  (s2) edge [] node [above] {} (s3)
  (s3) edge [] node [above] {} (s4)
  ;
\end{tikzpicture}
\caption{General form of a function arising in the $\utl$ contraction algorithm}
\label{fig:utlfunction}
\centering
\end{figure}
}
\begin{theorem}
$\utl$ path checking is in $\acc{1}$.
\end{theorem}

Same results apply to past temporal operators. Note that the construction works also when the negation is applied to arbitrary subformulae, and not only to propositions. Also note that $F_{[a,\infty)} p$ is downward monotone and the corresponding circuits are constructible in logarithmic space. Therefore, the above arguments apply to the more powerful logic $\utl_{\geq}$ obtained by allowing $F_{[a,\infty)}$ and $G_{[b,\infty)}$ operators.
To the best of our knowledge, $\utl_\geq$ is the most expressive
and powerful fragment of $\ltl$ with a sub-$\acd$ path-checking
problem.

\section{Conclusion}


The results obtained in this paper shed further light on the
complexity landscape of temporal-logic path-checking problems.
Several open questions however remain, the main one being to determine
the precise complexity of $\ltl$ path checking. In particular, there
has been no progress on the trivial $\ncc{1}$ lower bound over the
past ten years.
Furthermore, might it be possible to separate the complexity of $\ltl$ and
$\mtl$, or of these logics and their future-only fragment?

\textbf{Acknowledgments.} This research was financially supported by EPSRC.
\bibliographystyle{plain}
\bibliography{bib}

\begin{thebibliography}{10}

\bibitem{Artho:2005}
C.~Artho, H.~Barringer, A.~Goldberg, K.~Havelund, S.~Khurshid, M.~Lowry,
  C.~Pasareanu, G.~Rosu, K.~Sen, W.~Visser, and R.~Washington.
\newblock Combining test case generation and runtime verification.
\newblock {\em Theoretical Computer Science}, 336(2-3), 2005.

\bibitem{Barrington:1999}
D.A.M. Barrington, Chi-Jen Lu, P.B. Miltersen, and S.~Skyum.
\newblock On monotone planar circuits.
\newblock In {\em Proceedings of Fourteenth Annual IEEE Conference on
  Computational Complexity, 1999.}, 1999.

\bibitem{Chakraborty:2006}
T.~Chakraborty and S.~Datta.
\newblock One-input-face {MPCVP} is hard for {L}, but in {LogDCFL}.
\newblock In S.~Arun-Kumar and Naveen Garg, editors, {\em FSTTCS 2006:
  Foundations of Software Technology and Theoretical Computer Science}, volume
  4337 of {\em Lecture Notes in Computer Science}. 2006.

\bibitem{Demri:2002}
S.~Demri and Ph. Schnoebelen.
\newblock The complexity of propositional linear temporal logics in simple
  cases.
\newblock {\em Information and Computation}, 174(1):84 -- 103, 2002.

\bibitem{Finkbeiner:2004}
B.~Finkbeiner and H.~Sipma.
\newblock {Checking finite traces using alternating automata}.
\newblock {\em Formal Methods in System Design}, 24(2):101--127, October 2004.

\bibitem{Goldschlager:1977}
L.~M. Goldschlager.
\newblock The monotone and planar circuit value problems are log space complete
  for {P}.
\newblock {\em SIGACT News}, 9, July 1977.

\bibitem{Karp:1990}
R.~M. Karp and V.~Ramachandran.
\newblock Parallel algorithms for shared-memory machines.
\newblock In {\em Handbook of Theoretical Computer Science, Volume A:
  Algorithms and Complexity (A)}. 1990.

\bibitem{Kuhtz:2009}
L.~Kuhtz and B.~Finkbeiner.
\newblock {LTL} path checking is efficiently parallelizable.
\newblock In {\em Proceedings of the 36th Internatilonal Colloquium on
  Automata, Languages and Programming: Part II}, ICALP '09, 2009.

\bibitem{Kuhtz:2012}
L.~Kuhtz and B.~Finkbeiner.
\newblock Efficient parallel path checking for linear-time temporal logic with
  past and bounds.
\newblock {\em Logical Methods in Computer Science}, 8(4), 2012.

\bibitem{Limaye:2006}
N.~Limaye, M.~Mahajan, and J.~M.~N. Sarma.
\newblock Evaluating monotone circuits on cylinders, planes and tori.
\newblock In {\em Symposium on Theoretical Aspects of Computer Science: STACS},
  2006.

\bibitem{Maler:2004}
O.~Maler and D.~Nickovic.
\newblock {Monitoring temporal properties of continuous signals}.
\newblock {\em Formal Techniques, Modelling and Analysis of Timed and
  Fault-Tolerant Systems}, 2004.

\bibitem{Markey:03}
N.~Markey and Ph. Schnoebelen.
\newblock Model checking a path ({P}reliminary report).
\newblock In {\em Proceedings of the 14th International Conference on
  Concurrency Theory, Lecture Notes in Computer Science 2761}. Springer, 2003.

\bibitem{Younes:2002}
H.~L.~S. Younes and R.~G. Simmons.
\newblock Probabilistic verification of discrete event systems using acceptance
  sampling.
\newblock In {\em Proeedings of the 14th International Conference on Computer
  Aided Verification, volume 2404 of LNCS}. Springer, 2002.

\end{thebibliography}

\newpage

\appendix

\section{Proof of Lemma \ref{lemma:kijs}}
\begin{proof}

Note that, due to planarity of the underlying circuit, each $\aij$ depends on the consecutive block of gates one layer below and that $\aij$ and $\alpha_{i,j+1}$ share at most one predecessor.

\begin{compactitem}
\item {\boldmath $k_{i,j-1}<k_{i,j}$}
Note that for any $i,j$, the rightmost predecessor (successor) of $\alpha_{i,j}$ is always to the right in the planar embedding of to the rightmost predecessor (successor) of $\alpha_{i,j-1}$.  Let $\pi$ be the path constructed by taking the rightmost wires starting from $\alpha_{i,j}$. In particular, the wire from $\alpha_{i,j-1}$ to its rightmost predecessor is strictly to the left of $\pi$. Thus, $k_{i,j-1} < k_{i,j}$ as necessary. 

\item \textbf{\boldmath $k_{i,n_i}=k_{j,n_j} \leq m$ for every $i$ and $j$.} Since the rightmost gate can only be a predecessor of the rightmost gate one layer higher, it follows that the path $\pi$ is equal for all rightmost gates in all layers.

\item \textbf{\boldmath for every $i$ and $j=1,\ldots,n_i$ the intervals $v(i,j)$'s partition $[1,k_{i,n_i}]$.} By the first part, all $v(i,j)$'s in the same layer are disjoint. By construction, they cover $[1, k_{i,n_i}]$ entirely.

\item \textbf{\boldmath if $\alpha_{i+1,j}$ depends on $\alpha_{i,p},\alpha_{i,p+1},\ldots,\alpha_{i,q}$ then $$v(i+1,j)\subseteq \bigcup_{r=p,\ldots,q} v(i,r)$$ and $v(i+1,j)$ overlaps with each $v(i,r)$ for $p\leq r\leq q$.}
Construct the paths by taking the rightmost wires from $\alpha_{i+1,j-1}, \alpha_{i+1,j}, \alpha_{i,p}$ and $\alpha_{i,q}$ and denote them by $\pi_{j-1}, \pi, \pi_p$ and $\pi_q$ respectively. Notice that $\pi$ and $\pi_p$ coincide on the layers above $i+1$ and that $\pi$ and $\pi_q$ coincide on the layers below $i$. On the other hand, $\pi_q$ is always to the right of $\pi_p$. Thus, $k_{i,p} < k_{i+1,j} \leq k_{i,q}$. Also, $\pi_{j-1}$ is always to the left of $\pi_p$. Thus $k_{i+1,j-1} < k_{i,p}$ as required.

\item \textbf{each $k_{i,j}$ can be computed in $\logspace$.}
The algorithm needs to keep track of one counter: number of wires strictly to the left, and the current gate. By traversing the list of wires, it is easy to calculate the rightmost wire going up (down). By another pass through the list of wires, the counter is incremented by the wires to the left in the current layer. 

\end{compactitem}
\qed
\end{proof}

\section{Proof of Lemma \ref{lemma:phi}}

\begin{lemma}
\label{lemma:phiijs}
Fix $i>0, j$ and let $p\in\bb^n$ be any proposition that agrees with $r_{i-1}$ on $v(i,j)$. Then $\varphi_{i,j}(p)$ is constant on $v(i,j)$ and equals the output of $\alpha_{i,j}$ and $\varphi_{i,j}(p)$ is equal to $p$ elsewhere.
\end{lemma}

\begin{proof}
\textbf{\boldmath Case $\tau = \onec$}
Since $\bot \vee x = x$ and $\top \vee x = \top$, it holds that
$$(\chi_{l,r} \vee p)(i) = \left\{
  \begin{array}{l l}
    \top \vee p(i) = \top       & \quad \text{if $i \leq l \leq r$} \\
    \bot \vee p(i) = p(i)       & \quad \text{otherwise} \\
  \end{array} \right.
$$

\textbf{\boldmath Case $\tau = \zeroc$}
Similar to above, using $\top \wedge x = x$ and $\bot \wedge x = \bot$.

\textbf{\boldmath Case $\tau = \idc$}
Trivial as $\varphi_{i,j}$ is the identity \fcontext.

\textbf{\boldmath Case $\tau = \orc$}
Using the equality $r \U s = s \vee (r \wedge \X (r \U s))$, it follows that $(\chi_{l,r-1}\U p)(i) = p(i)$ for $i \not\in [l,r-1]$ as $\chi_{l,r-1}(i) = \bot$. In particular, $(\chi_{l,r-1}\U p)(r) = p(r)$. By induction, one can easily show that $(\chi_{l,r-1}\U p)(j) = p(j) \vee p(j+1) \vee \cdots \vee p(r)$ for $j \in [l,r]$.

Similarly, for any proposition $s$, it holds that
$$(\chi_{l+1,r} \since s)(j) = \left\{
  \begin{array}{l l}
    s(l) \vee s(l+1) \vee \cdots \vee s(j)   & \quad \text{if $j \in [l,r]$} \\
    s(j)                                     & \quad \text{otherwise} \\
  \end{array} \right.
$$

Putting the above equalities together, we obtain that
$$\varphi_{i,j}(p)(k)(\chi_{l+1,r} \since (\chi_{l,r-1} \U p))(k) = \left\{
  \begin{array}{l l}
    p(l) \vee p(l+1) \vee \cdots \vee p(r)   & \quad \text{if $k \in [l,r]$} \\
    p(k)                                     & \quad \text{otherwise} \\
  \end{array} \right.
$$

\textbf{Case $\tau = \andc$}
Dual to the above.\qed
\end{proof}

\begin{proof}[of Lemma \ref{lemma:phi}]
Let $\gamma_j = \varphi_{i,j} \circ \varphi_{i,2} \circ \cdots \circ \varphi_{i,n_i}$. We shall prove by downward induction that $\gamma_j(r_{i-1})$ equals $r_i$ on $v(i,j)\cup\cdots\cup v(i,n_i)$ and $r_{i-1}$ elsewhere.

\textbf{\boldmath Case $j = n_i$}
Immediate from Lemma~\ref{lemma:phiijs} applied to $\varphi_{i,n_i}$. 

\textbf{\boldmath Case $j < n_i$}
Let $\pi = \gamma_{j+1}(r_{i-1})$. Now, $v(i,t)$'s partition the trace, Lemma~\ref{lemma:kijs}, and, Lemma~\ref{lemma:phiijs}, the value of $\varphi_{i,j}$ on $v(i,j)$ depends only on the trace segment $v(i,j)$ and all other trace values are left unaffected. By the induction hypothesis, $\pi$ agrees with $r_{i-1}$ on $v(i,j)$. Thus, $\varphi_{i,j}(\pi)$ equals $r_i$ on $v(i,j)$ and $\pi$ everywhere else. By the induction, $\varphi_{i,j}(\pi)$ equals $r_i$ on $v(i,j)\cup\cdots\cup v(i,n_i)$ and $r_{i-1}$ elsewhere as required.

Finally, by induction, it is easily seen that $\varphi = (\psi_k \circ \psi_{k-1} \circ \cdots \circ \psi_1)(r_0) = r_n$. Thus $r_0\models\varphi$ if and only if $r_n(1)=\top$ which is precisely when the output gate of the circuit evaluates to true.\qed
\end{proof}

\section{Proof of Lemma \ref{lemma:circuit}}
One can easily show by induction on $q-p$ that $d_{p,q} = \varphi_p \vee \cdots \vee \varphi_q$. As noted in the main text,$$\pi,i \models s \U_I \varphi \text{ if and only if }\bigvee_{L_i\leq j\leq R_i} \pi,j\models \varphi \text{ is true.}$$
The result now immediately follows.

\section{Proof of Lemma \ref{lemma:utlcomposition}}
\begin{proof}
We show only how to compose $g:\mm\to\bb^n$ and a filter $f$ in $\acc{0}$. Other cases can be treated similarly. Let $h = f \circ g$ and for every admissible $k$ let $f_k$ be a Boolean variable which is true if and only if $f$ has offset $k$. Then for any input $v\in\mm$ and index $1\leq i\leq n$. The value of $h(v)(i) = (f(g(v)))(i)$. Now, depending on the offset and the pattern $p$ of $f$, the function $h(v)(i)$ evaluates to true if and only if

$$\bigvee_{k,1\leq i+k \leq n} (f_k \wedge ( ((p(i+k) = \top) \vee (p(i+k) = \idc \wedge g(i+k)) \vee (p(i+k) = \notc \wedge \neg g(i+k)))$$

is true where $p(i+k) = \top, \idc, \notc$ is a shorthand for the check that the bits encoding $p(i+k)$ encode the particular pattern. This is a constant depth polynomial size formula, thus a circuit, as necessary.\qed
\end{proof}

\end{document}